\documentclass[a4paper,USenglish,cleveref, autoref, thm-restate]{lipics-v2021}
\usepackage{tikz}
\usetikzlibrary{arrows.meta}
\def\2opt{2-Opt heuristic}

\title{The Approximation Ratio of the 2-Opt~Heuristic for the Euclidean Traveling Salesman Problem} %TODO Please add

\titlerunning{The Approximation Ratio of the 2-Opt~Heuristic for the Euclidean TSP} %TODO optional, please use if title is longer than one line

\author{Ulrich A.\ Brodowsky}{Pontsheide 20, 52076 Aachen, Germany}{}{}{}%TODO mandatory, please use full name; only 1 author per \author macro; first two parameters are mandatory, other parameters can be empty. Please provide at least the name of the affiliation and the country. The full address is optional

\author{Stefan Hougardy\footnote{corresponding author}}{Research Institute for Discrete Mathematics, University of Bonn, Lenn\'estr.~2, 53113 Bonn, Germany}{hougardy@or.uni-bonn.de}{https://orcid.org/0000-0001-8656-3418}{funded by the Deutsche Forschungsgemeinschaft (DFG, German Research Foundation) under Germany's Excellence Strategy --- EXC-2047/1 --- 390685813}

\authorrunning{U.\,A. Brodowsky and S. Hougardy} %TODO mandatory. First: Use abbreviated first/middle names. Second (only in severe cases): Use first author plus 'et al.'

\Copyright{Ulrich A.\ Brodowsky and Stefan Hougardy} %TODO mandatory, please use full first names. LIPIcs license is "CC-BY";  http://creativecommons.org/licenses/by/3.0/

 %TODO mandatory: Please choose ACM 2012 classifications from https://dl.acm.org/ccs/ccs_flat.cfm 

 %TODO mandatory; please add comma-separated list of keywords

\acknowledgements{We thank the anonymous referees for several useful comments.}%optional

\nolinenumbers %uncomment to disable line numbering
\hideLIPIcs

\EventEditors{Markus Bl\"{a}ser and Benjamin Monmege}
\EventNoEds{2}
\EventLongTitle{38th International Symposium on Theoretical Aspects of Computer Science (STACS 2021)}
\EventShortTitle{STACS 2021}
\EventAcronym{STACS}
\EventYear{2021}
\EventDate{March 16--19, 2021}
\EventLocation{Saarbr\"{u}cken, Germany (Virtual Conference)}
\EventLogo{}
\SeriesVolume{187}
\ArticleNo{41}

\begin{document}

\maketitle

%TODO mandatory: add short abstract of the document
\begin{abstract}
The \2opt is a simple improvement heuristic for the Traveling Salesman Problem.
It starts with an arbitrary tour and then repeatedly replaces two edges of the tour  
by two other edges, as long as this yields a shorter tour.
We will prove that for Euclidean Traveling Salesman Problems with $n$ cities 
the approximation ratio of the \2opt is $\Theta(\log n / \log \log n)$. 
This improves the upper bound of $O(\log n)$ given by Chandra, Karloff, and Tovey~\cite{CKT1999} in 1999.
\end{abstract}

\newpage

\section{Introduction}

The Traveling Salesman Problem (TSP) is one of the best studied problems in combinatorial optimization.
Given $n$ cities and their pairwise distances, the task is to find a shortest tour that visits each city exactly once.
This problem is NP-hard~\cite{GJ1979} and it is even hard to approximate to a factor that is polynomial in $n$~\cite{SG1976}.

In the \emph{Euclidean TSP}, the cities are points in $\mathbb{R}^2$ and the distance function is the Euclidean distance between 
the points. The Euclidean TSP is also NP-hard~\cite{Pap1977} but it allows a polynomial time approximation scheme~\cite{Aro1998,Mit1999}. 
Euclidean Traveling Salesman Problems often appear in practice and they are usually solved using some heuristics. 
One of the simplest of these heuristics is the \emph{\2opt}. It starts with an arbitrary tour and then repeatedly replaces two edges of the tour  
by two other edges, as long as this yields a shorter tour. The \2opt stops when no further improvement can be made this way. 
A tour that the \2opt cannot improve is called \emph{2-optimal}.

On real-world instances the \2opt achieves surprisingly good results (see e.g.\ Bentley~\cite{Ben1992}). 
Despite its simplicity the exact approximation ratio of the \2opt for Euclidean TSP is not known. 
In 1999, Chandra, Karloff, and Tovey~\cite{CKT1999} proved a lower bound of $c\cdot \frac{\log n}{\log \log n}$ 
for some constant $c > 0$ and an upper bound of $O(\log n)$ on the approximation ratio of the \2opt\ for 
Euclidean TSP. This leaves a gap of factor $O(\log \log n)$ between the best known upper and lower bound for
the approximation ratio of the \2opt\ for Euclidean TSP. 
Our main result closes this gap up to a constant factor:

\begin{theorem}
The approximation ratio of the \2opt\ for Euclidean TSP instances with $n$ points is $\Theta(\log n / \log \log n)$.
\label{thm:main}
\end{theorem}

\subparagraph*{Related Results.} 
On real-world Euclidean TSP instances it has been observed that the  2-Opt heuristic 
needs a sub-quadratic number of iterations until it reaches a local optimum~\cite{Ben1992}. However, 
there exist worst-case Euclidean TSP instances for which the 2-Opt heuristic may need an exponential number of 
iterations~\cite{ERV2014}.

For $n$ points embedded into the $d$-dimensional Euclidean space $\mathbb{R}^d$ for some constant $d>2$ the approximation
ratio of the 2-Opt heuristic is bounded by $O(\log n)$ from above~\cite{CKT1999} and by  $\Omega(\log n / \log \log n)$
from below~\cite{Zho2020}.  

The Euclidean TSP is a special case of the \emph{metric TSP}, i.e., the Traveling Salesman Problem
where the distance function satisfies the triangle inequality. 
The well-known algorithm of Christofides~\cite{Chr1976} and Serdjukov~\cite{Ser1978} achieves an approximation ratio of $3/2$
for the metric TSP. If one allows randomization then the recent algorithm of Karlin, Klein, and Oveis Gharan~\cite{KKO2020} 
slightly improves on this. For the metric TSP the 2-Opt heuristic 
has approximation ratio exactly $\sqrt{n/2}$~\cite{HZZ2020}. 
A very special case of the metric TSP is the 1-2-TSP. In this version all edge lengths
have to be~1 or~2. For the 1-2-TSP the approximation ratio of the 2-Opt heuristic is $3/2$~\cite{KMSV1998}.

For a constant $k > 2$ the 2-Opt heuristic naturally extends to the so called \emph{$k$-Opt heuristic} where 
in each iteration $k$ edges of a 
tour are replaced by $k$ other edges. For Euclidean TSP Zhong~\cite{Zho2020} has shown that $\Omega(\log n/ \log\log n)$ is
a lower bound for the $k$-Opt heuristic if $k$ is constant. Therefore, Theorem~\ref{thm:main} immediately implies:

\begin{corollary} For constant $k$ the 
approximation ratio of the $k$-Opt heuristic for Euclidean TSP instances with $n$ points is $\Theta(\log n / \log \log n)$.
\end{corollary}

\subparagraph*{Organization of the paper.}
The result of Chandra, Karloff, and Tovey~\cite[Theorem 4.4]{CKT1999} mentioned above shows that  
$\Omega(\log n / \log \log n)$ is a lower bound for the 2-Opt heuristic for Euclidean TSP. 
To prove Theorem~\ref{thm:main} it remains to prove the upper bound $O(\log n / \log \log n)$. 
We proceed as follows. First we will present in Section~\ref{sec:notation} some properties of Euclidean 2-optimal tours. 
In Section~\ref{sec:uncrossing} we will prove Theorem~\ref{thm:main} by reducing it to the special case 
where no intersections between the edges of an optimal tour and the edges of a 2-optimal tour exist. 
In this special case we will show that we can partition the edge set of a 2-optimal tour into five sets 
that are each in some sense orientation-preserving with respect to
an optimal tour. The main step then is to prove that for each of these five sets we can bound the total edge length by 
$O(\log n / \log \log n)$ times the length of an optimal tour. 
To achieve this we will relate optimal tours and subsets of the edge set of a 2-optimal tour to some weighted arborescences. 
This relation is studied in Section~\ref{sec:proofidea}.
For weighted arborescences we will provide in Section~\ref{sec:arborescence-lemmas} some bounds for the edge weights. 
These results then will allow us in Section~\ref{sec:proof} to finish the proof of Theorem~\ref{thm:main}.

\section{Euclidean TSP and 2-Optimal Tours}
\label{sec:notation}

An instance of the Euclidean TSP is a finite subset $V\subset \mathbb{R}^2$. 
The task is to find a polygon of shortest total edge length that contains all points of $V$.
Note that by our definition a Euclidean TSP instance cannot contain the same point multiple times.
In the following we will denote the cardinality of $V$ by $n$.

For our purpose it is often more convenient to state the Euclidean Traveling Salesman Problem as a problem on graphs. 
For a given point set $V$ of a Euclidean TSP instance we take a complete graph on the vertex set $V$, i.e., the graph $G=(V,E)$
where $E$ is the set of all $\frac12  n(n-1)$ possible edges on $V$. We assign the Euclidean  
distance between the vertices in $G$ by a function  $c:E(G) \to \mathbb{R}_{> 0}$.
A \emph{tour} in $G$ is a cycle that contains all the vertices of $G$.
The \emph{length} of a tour $T$ in $G$ is defined as $c(T) := \sum_{e\in E(T)} c(e)$.  
An \emph{optimal tour} is a tour of minimum length among the tours in $G$. 
Thus we can restate the Euclidean TSP as a problem in graphs: 
Given a complete graph $G=(V,E)$ on a point set $V\subset \mathbb{R}^2$ and a Euclidean distance function $c:E(G) \to \mathbb{R}_{> 0}$, 
find an optimal tour in $G$. 
Throughout this paper we will use the geometric definition of the Euclidean TSP and the graph-theoretic version of the
Euclidean TSP simultaneously. Thus, a tour for a Euclidean TSP instance $V\subseteq \mathbb{R}^2$ can be viewed as a polygon in  $\mathbb{R}^2$ as well as
a cycle in a complete graph on the vertex set $V$ with Euclidean distance function.

Let $c:E(G)\to\mathbb{R}_{>0}$ be a weight function for the edges of some graph $G=(V,E)$. 
To simplify notation, we will denote the weight of an edge $\{x,y\} \in E(G)$ simply by $c(x,y)$ instead 
of the more cumbersome notation $c(\{x,y\})$. For subsets $F\subseteq E(G)$ we define
$c(F) := \sum_{e\in F} c(e)$. We extend this definition to subgraphs $H$ of $G$ by setting $c(H) := c(E(H))$.

The distance function $c$ of a Euclidean TSP instance $G=(V,E)$ satisfies the triangle inequality. Therefore we have 
for any set of three vertices $x, y, z\in V(G)$:
\begin{equation}
c(x,y) ~+~ c(y, z) ~ ~\ge~ ~ c(x, z).
\end{equation}

The \2opt repeatedly replaces two edges from the tour by two other edges such that the resulting tour is shorter.
Given a tour $T$ and two edges $\{a,b\}$ and $\{x,y\}$ in $T$, there are two possibilities to replace these two edges by two other edges.
Either we can choose the pair $\{a,x\}$ and $\{b,y\}$ or we can choose the pair $\{a,y\}$ and $\{b,x\}$. Exactly one of these two pairs
will result in a tour again. Without knowing the other edges of $T$, we cannot decide which of the two possibilities is the correct one.
Therefore, we will assume in the following that the tour $T$ is an \emph{oriented} cycle, i.e., the edges of $T$ have an orientation 
such that each vertex has exactly one incoming and one outgoing edge. Using this convention, there is only one possibility to
exchange a pair of edges such that the new edge set is a tour again: two directed edges $(a,b)$ and $(x,y)$ have to be replaced by
the edges $(a,x)$ and $(b,y)$. Note that to obtain an oriented cycle again, one has to reverse the direction of the segment between $b$ and $x$, 
see Figure~\ref{fig:2-Opt}.

\begin{figure}[t]
\centering
\begin{tikzpicture}[scale=0.25]
\tikzstyle{vertex}=[blue,circle,fill, minimum size=5, inner sep=0]
\tikzstyle{arrow}=[Straight Barb[length=1mm]]

\node[vertex, label=above:$$] (P1)  at (17, 0) {};
\node[vertex, label=above:$b$] (P2)  at (20, 8) {};
\node[vertex, label=above:$$] (P3)  at ( 6, 2) {};
\node[vertex, label=above:$~y$] (P4)  at (15, 4) {};
\node[vertex, label=above:$$] (P5)  at (22, 3) {};
\node[vertex, label=above:$x$] (P6)  at ( 8, 7) {};
\node[vertex, label=above:$$] (P7)  at ( 0, 4) {};
\node[vertex, label=above:$$] (P8)  at (12, 2) {};
\node[vertex, label=above:$$] (P9)  at ( 2, 9) {};
\node[vertex, label=above:$a$] (P10) at (13, 9) {};

\draw[-{Straight Barb[length=1mm]},  red, line width=1]   (P10) to (P2);
\draw[-{Straight Barb[length=1mm]},  red, line width=1]   (P6)  to (P4);
\draw[-{Straight Barb[length=1mm]},  line width=0.4] (P4)  to (P10);
\draw[-{Straight Barb[length=1mm]},  line width=0.4] (P2)  to (P5);
\draw[-{Straight Barb[length=1mm]},  line width=0.4] (P5)  to (P1);
\draw[-{Straight Barb[length=1mm]},  line width=0.4] (P1)  to (P8);
\draw[-{Straight Barb[length=1mm]},  line width=0.4] (P8)  to (P3);
\draw[-{Straight Barb[length=1mm]},  line width=0.4] (P3)  to (P7);
\draw[-{Straight Barb[length=1mm]},  line width=0.4] (P7)  to (P9);
\draw[-{Straight Barb[length=1mm]},  line width=0.4] (P9)  to (P6);

\begin{scope}[shift={(32,0)}]
\node[vertex, label=above:$$] (P1)  at (17, 0) {};
\node[vertex, label=above:$b$] (P2)  at (20, 8) {};
\node[vertex, label=above:$$] (P3)  at ( 6, 2) {};
\node[vertex, label=above:$~y$] (P4)  at (15, 4) {};
\node[vertex, label=above:$$] (P5)  at (22, 3) {};
\node[vertex, label=above:$x$] (P6)  at ( 8, 7) {};
\node[vertex, label=above:$$] (P7)  at ( 0, 4) {};
\node[vertex, label=above:$$] (P8)  at (12, 2) {};
\node[vertex, label=above:$$] (P9)  at ( 2, 9) {};
\node[vertex, label=above:$a$] (P10) at (13, 9) {};

\draw[-{Straight Barb[length=1mm]},  red, line width=1]   (P10) to (P6);
\draw[-{Straight Barb[length=1mm]},  red, line width=1]   (P2)  to (P4);
\draw[-{Straight Barb[length=1mm]},  line width=0.4] (P4)  to (P10);
\draw[-{Straight Barb[length=1mm]},  line width=0.4] (P5)  to (P2);
\draw[-{Straight Barb[length=1mm]},  line width=0.4] (P1)  to (P5);
\draw[-{Straight Barb[length=1mm]},  line width=0.4] (P8)  to (P1);
\draw[-{Straight Barb[length=1mm]},  line width=0.4] (P3)  to (P8);
\draw[-{Straight Barb[length=1mm]},  line width=0.4] (P7)  to (P3);
\draw[-{Straight Barb[length=1mm]},  line width=0.4] (P9)  to (P7);
\draw[-{Straight Barb[length=1mm]},  line width=0.4] (P6)  to (P9);

\end{scope}

\end{tikzpicture}
\caption{An oriented TSP tour (left) and the tour obtained after replacing the edges  $(a,b)$ and $(x,y)$ with the edges $(a,x)$ and $(b,y)$ (right).
The orientation of the tour segment between the vertices $b$ and $x$ has been reversed in the new tour.}
\label{fig:2-Opt}
\end{figure}
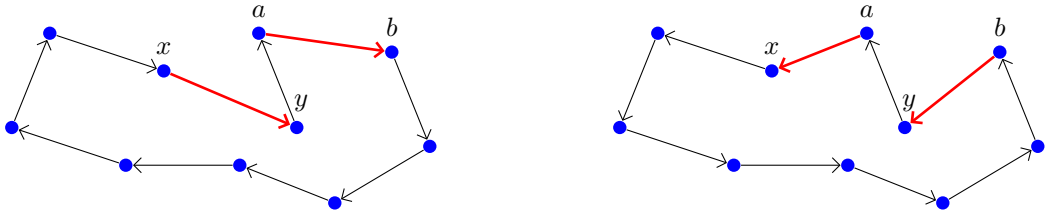

A TSP tour $T$ is called \emph{2-optimal} if for any two edges $(a,b)$ and $(x,y)$ of $T$ we have
\begin{equation}
 c(a,x) + c(b,y) ~\ge~ c(a,b) + c(x,y) 
 \label{eq:2-optimality}
\end{equation}
We call inequality~(\ref{eq:2-optimality}) the \emph{2-optimality condition}.

If $(a,b)$ and $(x,y)$ are two edges in a tour $T$ that violate the 2-optimality condition, i.e., they satisfy the inequality 
$ c(a,x) + c(b,y) < c(a,b) + c(x,y)$, 
then we can replace the edges $(a,b)$ and $(x,y)$ in $T$ by  the edges $(a,x)$ and $(b,y)$ and get a strictly shorter tour.
We call this operation of replacing the edges $(a,b)$ and $(x,y)$ in $T$ by  the edges $(a,x)$ and $(b,y)$
an \emph{improving 2-move}. Thus, the \2opt can be formulated as follows: \medskip

\framebox{\parbox{10cm}{
\noindent
{\bfseries 2-Opt Heuristic} ($V\subseteq \mathbb{R}^2$)\\[2mm]
1~ start with an arbitrary tour $T$ for $V$\\
2~ \texttt{while} there exists an improving 2-move in $T$ \\
3~ ~~~~  perform an improving 2-move\\
4~ \texttt{output} $T$}}\medskip

We call a Euclidean TSP instance  $V\subset \mathbb{R}^2$ \emph{degenerate} if there exists a line in $\mathbb{R}^2$ that
contains all points of $V$. Otherwise we call the instance \emph{non-degenerate}. 

It is easily seen that in a degenerate Euclidean TSP instance a 2-optimal tour is also an optimal tour:

\begin{proposition}
In a degenerate Euclidean TSP instance a 2-optimal tour is an optimal tour.
\label{prop:degenerate-case}
\end{proposition}
\begin{proof}
Let $V\subseteq \mathbb{R}^2$ be a degenerate Euclidean TSP instance. 
Let $c:V\times V \to \mathbb{R}$ be the Euclidean distance between two points in $V$.
Then there exist two points $a, b\in V$ such that the straight line segment $S$ from $a$ to $b$ contains all points of $V$. 
The length of an optimal TSP tour for $V$ is $2\cdot c(a,b)$. Assume there exists a $2$-optimal TSP tour $T$ that is not optimal. Orient the tour $T$. Then there must exist a point 
in $S\setminus V$ that is contained in at least three edges of the tour $T$ and therefore there must exist a point in $S\setminus V$ 
that is contained in two edges $(v,w)$ and $(x,y)$ of $T$ that are oriented in the same direction. This contradicts the 2-optimality of $T$ as 
$c(u,x) + c(w,y) < c(v,w) + c(x,y)$.   
\end{proof}

Because of Proposition~\ref{prop:degenerate-case} we may assume in the following that we have a non-degenerate Euclidean TSP instance. 

Let $T$ be a tour in a Euclidean TSP instance. Each edge of $T$ corresponds to a closed line segment in $\mathbb{R}^2$. 
A tour in a Euclidean TSP instance is called \emph{simple} if no two edges of the tour intersect in a point that lies in the interior of at least one
of the two corresponding line segments. 
For 2-optimal tours in Euclidean TSP instances we have the following simple but very important result. 

\begin{lemma}(Flood~\cite{Flo1956}) In a non-degenerate Euclidean TSP instance a 2-optimal tour is simple.
\label{lemma:nocrossing}
\end{lemma}

\subsection{Crossing-Free Pairs of Tours}
\label{sec:uncrossing}

Let $T$ be an optimal tour and $T'$ be a 2-optimal tour in a non-degenerate Euclidean TSP instance. 
By Lemma~\ref{lemma:nocrossing} we know that both tours are simple. In the following we want to justify a much stronger assumption.
Two edges $e\in E(T)$ and $f\in E(T')$ \emph{cross} if $e$ and $f$ intersect in exactly one point in $\mathbb{R}^2$ 
and this point is in the interior of both line segments.  
We say that two tours $T$ and $T'$ are \emph{crossing-free} if there does not exist a pair of crossing edges.
See Figure~\ref{fig:crossingEdges} for an example of an optimal tour and a 2-optimal tour that have three crossing pairs of edges.

To prove Theorem~\ref{thm:main} it will be enough to prove it for the special case of 
crossing-free tours: 

\begin{theorem} Let $V\subseteq \mathbb{R}^2$ with $|V| = n$ be a non-degenerate Euclidean TSP instance, 
$T$ an optimal tour for $V$ and $S$ a 2-optimal tour for $V$.
If $T$ and $S$ are crossing-free then the length of $S$ is bounded by $O(\log n / \log \log n)$ times the length of $T$.
\label{thm:main-crossingfree}
\end{theorem}

\begin{figure}[t]
\centering
\begin{tikzpicture}[scale=0.43, fill=gray]

\coordinate  (P1) at (11, 12);
\coordinate  (P2) at ( 9,  3);
\coordinate  (P3) at ( 0,  7);
\coordinate  (P4) at (11,  4);
\coordinate  (P5) at ( 5, 13);
\coordinate  (P6) at ( 3, 15);
\coordinate  (P7) at ( 4, 12);
\coordinate  (P8) at (11,  2);
\coordinate  (P9) at ( 6,  9);
\coordinate (P10) at (10,  5);
\coordinate (P11) at (12, 14);
\coordinate (P12) at (14, 13);

\filldraw[fill=black!10!white, draw=red, line width = 1.5] 
(P1) -- (P11) -- (P12) -- (P10) -- (P4) -- (P8) -- (P2) -- (P9) -- (P3) -- (P7) -- (P6) -- (P5) -- cycle;

\draw[color=green!80!black, dash pattern = on 1.5mm off 1.5mm, line width = 1.5] 
(P1) -- (P4) -- (P8) -- (P2) -- (P10) -- (P3) -- (P7) -- (P6) -- (P5) -- (P9) -- (P11) -- (P12) -- cycle;

% solid arrowed line  
\tikzstyle{sal}=[-{Straight Barb[length=1.5mm]}, color=blue!70!white, line width = 1.5] 

\foreach \i in {1,...,12}
  \fill[black] (P\i) circle (2.5mm);

\begin{scope}[shift={(17,0)}]
\coordinate  (P1) at (11, 12);
\coordinate  (P2) at ( 9,  3);
\coordinate  (P3) at ( 0,  7);
\coordinate  (P4) at (11,  4);
\coordinate  (P5) at ( 5, 13);
\coordinate  (P6) at ( 3, 15);
\coordinate  (P7) at ( 4, 12);
\coordinate  (P8) at (11,  2);
\coordinate  (P9) at ( 6,  9);
\coordinate (P10) at (10,  5);
\coordinate (P11) at (12, 14);
\coordinate (P12) at (14, 13);

\filldraw[fill=black!10!white, draw=red, line width = 1.5] 
(P1) -- (P11) -- (P12) -- (P10) -- (P4) -- (P8) -- (P2) -- (P9) -- (P3) -- (P7) -- (P6) -- (P5) -- cycle;

\draw[color=green!80!black, dash pattern = on 1.5mm off 1.5mm, line width = 1.5] 
(P1) -- (P4) -- (P8) -- (P2) -- (P10) -- (P3) -- (P7) -- (P6) -- (P5) -- (P9) -- (P11) -- (P12) -- cycle;

\coordinate (P13) at (11.000000,  7.000000);
\coordinate (P14) at ( 9.833333, 12.194444); 
\coordinate (P15) at ( 7.777777,  5.444444);

\foreach \i in {1,...,12}
  \fill[black] (P\i) circle (2.5mm);

\foreach \i in {13,...,15}
  \fill[color=blue!70!white] (P\i) circle (2.5mm);

\end{scope}

\end{tikzpicture}
\caption{A Euclidean TSP instance with an optimal tour (red edges) and a 2-optimal tour (dashed green edges). Both tours shown in the left picture 
are simple but there are three pairs of crossing edges. The tours can be made crossing-free by adding three vertices (blue points in the right picture) to the instance.}
\label{fig:crossingEdges}
\end{figure}
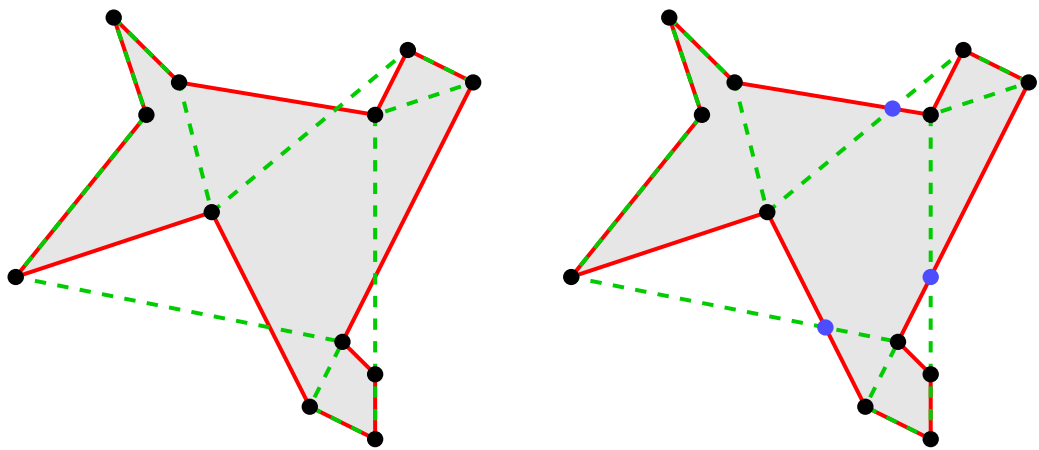

The proof of Theorem~\ref{thm:main-crossingfree} will be presented in Section~\ref{sec:proof}. Here we show how Theorem~\ref{thm:main-crossingfree}
allows to prove Theorem~\ref{thm:main}. For this we describe a method to transform a pair of tours into a crossing-free pair of tours.

Let $V\subseteq \mathbb{R}^2$ be a Euclidean TSP instance and $T$ be a tour for $V$. We say that $V'\subseteq \mathbb{R}^2$ 
is a \emph{subdivision} for $(V,T)$ if $V \subset V'$ and $V'$ is a subset of the polygon $T$.  The set $V'$ \emph{induces} a new tour $T'$ which results from 
the tour $T$ by subdividing the edges by points in $V'\setminus V$. Note that $T$ and $T'$ constitute the same polygon.
Therefore we have:

\begin{proposition}
Let $V\subseteq \mathbb{R}^2$ be a Euclidean TSP instance and $T$ be an optimal tour. 
If $V'$ is a subdivision  for $(V,T)$ then the tour $T'$ induced by $V'$ is an optimal tour for $V'$.
\label{prop:optimality-preserved}
\end{proposition}

Subdividing a tour not only preserves the optimality but it also preserves the 2-optimality:

\begin{lemma}
Let $V\subseteq \mathbb{R}^2$ be a Euclidean TSP instance and $T$ be a 2-optimal tour. 
If $V'$ is a subdivision  for $(V,T)$ then the tour $T'$ induced by $V'$ is a 2-optimal tour for $V'$.
\label{lemma:2-optimality-preserved}
\end{lemma}

\begin{proof}
Let us assume that the tour $T'$ is oriented and that $(x',y')$ and $(a',b')$ are two edges of $T'$. We have to prove that 
these two edges satisfy the 2-optimality condition~(\ref{eq:2-optimality}). 
As $T'$ is a subdivision of the 2-optimal tour $T$ we know that there exist edges $(x,y)$ and $(a,b)$ in $T$ such that
the line segment $a'b'$ is contained in the line segment $ab$ and the line segment $x'y'$ is contained in the line segment 
$xy$. The 2-optimality of $T$ implies
\[ c(a,b) + c(x,y) ~\le~ c(a,x) + c(b,y) .\]
Using this inequality and the triangle inequality we get:
\begin{eqnarray*}
c(a',b') + c(x', y') & =   & c(a,b) - c(a,a') - c(b,b') + c(x,y) - c(x,x') - c(y,y') \\
                     & \le & c(a,x) - c(a,a') - c(x,x') + c(b,y) - c(b,b') - c(y,y') \\
                     & \le & c(a',x') + c(b',y') .
\end{eqnarray*}             
\end{proof}

Now we are able to reduce Theorem~\ref{thm:main} to Theorem~\ref{thm:main-crossingfree}:

\noindent
\textit{Proof of Theorem~\ref{thm:main}:~}
Let $V\subseteq \mathbb{R}^2$ be a Euclidean TSP instance with $|V| = n$,  $T$ be an optimal tour for $V$ and
$S$ be a 2-optimal tour for $V$. 
By Proposition~\ref{prop:degenerate-case} we may assume that $V$ is non-degenerate.
Let $V'\subseteq \mathbb{R}^2$ be the set of points obtained by adding to $V$ all
crossings between pairs of edges in $T$ and $S$. 
Denote the cardinality of $V'$ by $n'$. 
Let $T'$ and $S'$ be the tours induced by $V'$ for $T$ and $S$. 
Then by Proposition~\ref{prop:optimality-preserved}
and by Lemma~\ref{lemma:2-optimality-preserved} we know that $T'$ has the same length as $T$ and is an optimal tour for $V'$ 
and $S'$ has the same length as $S$ and is a 2-optimal tour for $V'$.
Now Theorem~\ref{thm:main-crossingfree} implies that the length of $S$ is at most 
$O(\log n' / \log \log n')$ times the length of $T$. It remains to observe that there can be at most $O(n^2)$ crossings between 
edges in $T$ and $S$ and therefore $O(\log n' / \log \log n') = O(\log n / \log \log n)$.
\qed
\medskip

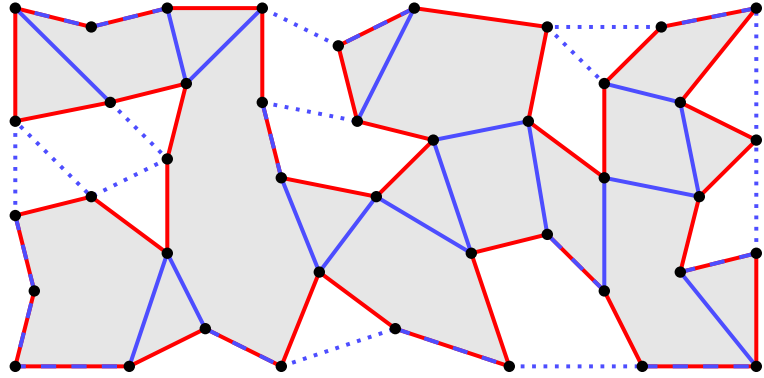
\begin{figure}[t]
\centering
\begin{tikzpicture}[scale=0.25, fill=gray]

\coordinate  (P1) at (32,  5);
\coordinate  (P2) at (18, 18);
\coordinate  (P3) at (27,  1);
\coordinate  (P4) at (15,  1);
\coordinate  (P5) at (23, 13);
\coordinate  (P6) at (40, 13);
\coordinate  (P7) at (34,  1);
\coordinate  (P8) at (15, 11);
\coordinate  (P9) at (11,  3);
\coordinate (P10) at (32, 11);
\coordinate (P11) at (35, 19);
\coordinate (P12) at ( 5, 10);
\coordinate (P13) at (21,  3);
\coordinate (P14) at (29, 19);
\coordinate (P15) at (25,  7);
\coordinate (P16) at (40,  7);
\coordinate (P17) at (36, 15);
\coordinate (P18) at (10, 16);
\coordinate (P19) at ( 1, 20);
\coordinate (P20) at (32, 16);
\coordinate (P21) at ( 9,  7);
\coordinate (P22) at (28, 14);
\coordinate (P23) at (36,  6);
\coordinate (P24) at ( 2,  5);
\coordinate (P25) at ( 7,  1);
\coordinate (P26) at (40,  1);
\coordinate (P27) at ( 1, 14);
\coordinate (P28) at ( 9, 20);
\coordinate (P29) at (20, 10);
\coordinate (P30) at (17,  6);
\coordinate (P31) at ( 6, 15);
\coordinate (P32) at (14, 15);
\coordinate (P33) at (22, 20);
\coordinate (P34) at (29,  8);
\coordinate (P35) at ( 1,  9);
\coordinate (P36) at (40, 20);
\coordinate (P37) at ( 9, 12);
\coordinate (P38) at (14, 20);
\coordinate (P39) at (37, 10);
\coordinate (P40) at ( 1,  1);
\coordinate (P41) at (19, 14);
\coordinate (P42) at ( 5, 19);

\filldraw[fill=black!10!white, draw=red, line width = 1.5] 
(P26) -- (P16) -- (P23) -- (P39) --  (P6) -- (P17) -- (P36) -- (P11) -- (P20) -- (P10) -- 
(P22) -- (P14) -- (P33) --  (P2) -- (P41) --  (P5) -- (P29) --  (P8) -- (P32) -- (P38) -- 
(P28) -- (P42) -- (P19) -- (P27) -- (P31) -- (P18) -- (P37) -- (P21) -- (P12) -- (P35) -- 
(P24) -- (P40) -- (P25) --  (P9) --  (P4) -- (P30) -- (P13) --  (P3) -- (P15) -- (P34) -- 
 (P1) --  (P7) -- cycle;

\draw[dash pattern = on 2mm off 2mm, color=blue!70!white, line width = 1.5] 
(P23) -- (P16)    (P36) -- (P11)     (P1) -- (P34)     (P8) -- (P32)    (P33) --  (P2) 
(P28) -- (P42) -- (P19)    (P35) -- (P24) -- (P40) -- (P25)     (P9) --  (P4)    (P13) -- 
 (P3)     (P7) -- (P26);

\draw[dash pattern = on 0.5mm off 1mm, color=blue!70!white, line width = 1.5] 
(P16) --  (P6) -- (P36)    (P11) -- (P14) -- (P20)    (P32) -- (P41)     (P2) -- (P38) 
(P31) -- (P37) -- (P12) -- (P27) -- (P35)     (P4) -- (P13)     (P3) --  (P7);

\draw[color=blue!70!white, line width = 1.5] 
(P26) -- (P23)    (P20) -- (P17) -- (P39) -- (P10) --  (P1)    (P34) -- (P22) --  (P5) -- 
(P15) -- (P29) -- (P30) --  (P8)    (P41) -- (P33)    (P38) -- (P18) -- (P28)    (P19) -- 
(P31)    (P25) -- (P21) --  (P9);

\foreach \i in {1,...,42}
  \fill[black] (P\i) circle (3mm);

\end{tikzpicture}
\caption{A Euclidean TSP instance with an optimal tour $T$ (red edges) and a 2-optimal tour 
(blue edges) that are crossing-free. The edges of the 2-optimal tour are partitioned into the 
edges lying in the interior of $T$ (solid blue lines), the edges that lie in the exterior of $T$
(dotted blue lines), and the edges that are part of $T$ (dashed blue edges).} 
\label{fig:crossingfreetours}
\end{figure}

\subsection{Partitioning the Edge Set of a 2-Optimal Tour}
\label{sec:edge-partition}

Let $V \subseteq\mathbb{R}^2$ be a non-degenerate Euclidean TSP instance, $T$ be an optimal tour and $S$ be a 2-optimal tour such that
$S$ and $T$ are crossing-free. As $S$ and $T$ are simple polygons and $S$ and $T$ are crossing-free we can partition the edge set of $S$
into three sets $S_1$, $S_2$, and $S_3$ such that all edges of $S_1$ lie in the interior of $T$, all edges of $S_2$ lie in the exterior of $T$
and all edges of $S_3$ are contained in $T$ (see Figure~\ref{fig:crossingfreetours}). 
More precisely, an edge $\{a, b\} \in S$ belongs to $S_1$ resp.\ $S_2$ if the corresponding open line segment $ab$ completely
lies in the interior resp.\ exterior of the polygon $T$. The set $S_3$ contains all the edges of $S$ that are subsets 
of the polygon $T$.

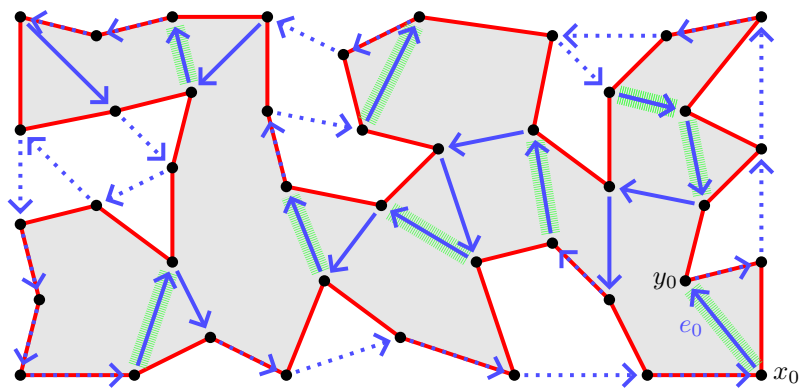
\begin{figure}[t]
\centering
\begin{tikzpicture}[scale=0.25, fill=gray]

\coordinate  (P1) at (32,  5);
\coordinate  (P2) at (18, 18);
\coordinate  (P3) at (27,  1);
\coordinate  (P4) at (15,  1);
\coordinate  (P5) at (23, 13);
\coordinate  (P6) at (40, 13);
\coordinate  (P7) at (34,  1);
\coordinate  (P8) at (15, 11);
\coordinate  (P9) at (11,  3);
\coordinate (P10) at (32, 11);
\coordinate (P11) at (35, 19);
\coordinate (P12) at ( 5, 10);
\coordinate (P13) at (21,  3);
\coordinate (P14) at (29, 19);
\coordinate (P15) at (25,  7);
\coordinate (P16) at (40,  7);
\coordinate (P17) at (36, 15);
\coordinate (P18) at (10, 16);
\coordinate (P19) at ( 1, 20);
\coordinate (P20) at (32, 16);
\coordinate (P21) at ( 9,  7);
\coordinate (P22) at (28, 14);
\coordinate (P23) at (36,  6);
\coordinate (P24) at ( 2,  5);
\coordinate (P25) at ( 7,  1);
\coordinate (P26) at (40,  1);
\coordinate (P27) at ( 1, 14);
\coordinate (P28) at ( 9, 20);
\coordinate (P29) at (20, 10);
\coordinate (P30) at (17,  6);
\coordinate (P31) at ( 6, 15);
\coordinate (P32) at (14, 15);
\coordinate (P33) at (22, 20);
\coordinate (P34) at (29,  8);
\coordinate (P35) at ( 1,  9);
\coordinate (P36) at (40, 20);
\coordinate (P37) at ( 9, 12);
\coordinate (P38) at (14, 20);
\coordinate (P39) at (37, 10);
\coordinate (P40) at ( 1,  1);
\coordinate (P41) at (19, 14);
\coordinate (P42) at ( 5, 19);

\filldraw[fill=black!10!white, draw=red, line width = 1.5] 
(P26) -- (P16) -- (P23) -- (P39) --  (P6) -- (P17) -- (P36) -- (P11) -- (P20) -- (P10) -- 
(P22) -- (P14) -- (P33) --  (P2) -- (P41) --  (P5) -- (P29) --  (P8) -- (P32) -- (P38) -- 
(P28) -- (P42) -- (P19) -- (P27) -- (P31) -- (P18) -- (P37) -- (P21) -- (P12) -- (P35) -- 
(P24) -- (P40) -- (P25) --  (P9) --  (P4) -- (P30) -- (P13) --  (P3) -- (P15) -- (P34) -- 
 (P1) --  (P7) -- cycle;

\foreach \i in {1,...,42}
  \node[circle, minimum size = 2.5mm, inner sep = 0mm] (N\i) at (P\i) {};

\node[circle, minimum size = 2.5mm, inner sep = 0mm, label=right:\hspace*{-1mm}$x_0$] (N26) at (P26) {};
\node[circle, minimum size = 2.5mm, inner sep = 0mm, label=left:$y_0$\hspace*{-1.6mm}]  (N23) at (P23) {};

%\tikzstyle{arrow}=[Straight Barb[length=2mm]] 

% solid arrowed line  
\tikzstyle{sal}=[-{Straight Barb[length=1.5mm]}, color=blue!70!white, line width = 1.5] 

% dotted arrowed line  
\tikzstyle{dal}=[dash pattern = on 0.5mm off 1mm, -{Straight Barb[length=1.5mm]}, color=blue!70!white, line width = 1.5] 

% highlighted line  
\tikzstyle{hl}=[dash pattern = on 0.1mm off 0.2mm, color=green!70!white, line width = 7] 

\draw[hl] (N26)  -- (N23); 
\draw[hl]  (N20) -- (N17); 
\draw[hl]  (N20) -- (N17); 
\draw[hl]  (N17) -- (N39); 
\draw[hl]  (N34) -- (N22); 
\draw[hl]  (N15) -- (N29); 
\draw[hl]  (N30) --  (N8); 
\draw[hl]  (N41) -- (N33); 
\draw[hl]  (N18) -- (N28); 
\draw[hl]  (N25) -- (N21); 

\draw[sal] (N26)  to node[left] {$e_0~$} (N23);
\draw[dal] (N23) -- (N16);
\draw[dal] (N16) -- (N6);
\draw[dal]  (N6) -- (N36);
\draw[dal] (N36) -- (N11);
\draw[dal] (N11) -- (N14);
\draw[dal] (N14) -- (N20);
\draw[sal] (N20) -- (N17);
\draw[sal] (N20) -- (N17);
\draw[sal] (N17) -- (N39);
\draw[sal] (N39) -- (N10);
\draw[sal] (N10) --  (N1);
\draw[dal]  (N1) -- (N34);
\draw[sal] (N34) -- (N22);
\draw[sal] (N22) --  (N5);
\draw[sal]  (N5) -- (N15);
\draw[sal] (N15) -- (N29);
\draw[sal] (N29) -- (N30);
\draw[sal] (N30) --  (N8);
\draw[dal]  (N8) -- (N32);
\draw[dal] (N32) -- (N41);
\draw[sal] (N41) -- (N33);
\draw[dal] (N33) --  (N2);
\draw[dal]  (N2) -- (N38);
\draw[sal] (N38) -- (N18);
\draw[sal] (N18) -- (N28);
\draw[dal] (N28) -- (N42);
\draw[dal] (N42) -- (N19);
\draw[sal] (N19) -- (N31);
\draw[dal] (N31) -- (N37);
\draw[dal] (N37) -- (N12);
\draw[dal] (N12) -- (N27);
\draw[dal] (N27) -- (N35);
\draw[dal] (N35) -- (N24);
\draw[dal] (N24) -- (N40);
\draw[dal] (N40) -- (N25);
\draw[sal] (N25) -- (N21);
\draw[sal] (N21) --  (N9);
\draw[dal]  (N9) --  (N4);
\draw[dal]  (N4) -- (N13);
\draw[dal] (N13) --  (N3);
\draw[dal]  (N3) --  (N7);
\draw[dal]  (N7) -- (N26);

\foreach \i in {1,...,42}
  \fill[black] (P\i) circle (3mm);

\end{tikzpicture}
\caption{The edge $e_0=(x_0,y_0)$ of the 2-optimal tour (blue edges) defines the set $S_1'$ of all green marked edges in the interior of the optimal tour $T$ (red edges).} 
\label{fig:edge-partition}
\end{figure}

By definition and because of Proposition~\ref{prop:degenerate-case} and Lemma~\ref{lemma:nocrossing} we know that the edges in $S_3$ are at most as long as the edges in $T$. 
To bound the total length of all edges in $S_1$ in terms of the length of $T$ we proceed as follows: Fix some 
orientation of
the tour $S$. We may assume that $S_1$ contains at least two edges as otherwise by the triangle inequality the length of $T$ is an upper bound for the
length of the edges in $S_1$. 
Choose an edge $e_0=(x_0,y_0)$ from $S_1$ such that one of the two $x_0$-$y_0$-paths in $T$ does not contain in its interior the endpoints of any
other edge in $S_1$ (see Figure~\ref{fig:edge-partition}).  

Let $T_{[x_0,y_0]}$ be the $x_0$-$y_0$-path in $T$ that contains the endpoints of all other edges in $S_1$. The path $T_{[x_0,y_0]}$ is unique if 
we assume $|S_1|\ge 2$. 
Then we define the set $S_1'$ to contain all edges from $S_1$ that are ``compatible'' with $T_{[x_0,y_0]}$ in the following sense:
\[ S_1' := \{(a,b) \in S_1 : \mbox{ the } x_0\mbox{-}b\mbox{-}\mbox{path in } T_{[x_0,y_0]} \mbox{ contains } a\} . \]
In Figure~\ref{fig:edge-partition} for the chosen edge $e_0=(x_0,y_0)$ the edges in $S_1'$ are marked green.

All edges in $S_1$ that are oriented in the ``wrong'' way with respect to $T_{[x_0,y_0]}$ define the set $S_1''$, i.e., we have
$S_1'' := S_1 \setminus S_1'$. Similarly we can define sets $S_2'$ and $S_2''$ with respect to some edge $f\in S$ that lies in the
exterior of $T$. We want to prove that for each of the four sets $S_1', S_1'', S_2', S_2''$ we can bound the total length of all edges by 
 $O(\log n / \log \log n)$ times the length of $T$. To achieve this we will reduce the problem to a problem in 
weighted arborescences.

\section{Arborescences and Pairs of Tours}
\label{sec:proofidea}

In this section we will explain why bounding the length of a 2-optimal tour
reduces to some problem in weighted arborescences. We start by giving an informal description of the idea. 

Let $T$ be an optimal tour for a non-degenerate Euclidean TSP instance $V\subseteq \mathbb{R}^2$ and let $S$ be a 2-optimal tour such that $S$ and $T$ are crossing-free. 
Then $T$ together with the edge set $S_1'$ as defined in Section~\ref{sec:edge-partition} is a plane graph. 
Each region of this plane graph is bounded by edges in $E(T) \cup S_1'$. The boundary of each region is a cycle. Because of the
triangle inequality we can bound the length of each edge in a cycle by the sum of the lengths of all other edges in this cycle. 
This way we get a bound for the length of each edge in $S_1'$ which we call the \emph{combined triangle inequality} as
it arises by applying the triangle inequality to edges from both tours $T$ and $S$.

From the boundaries of the regions of the plane graph we can derive
another type of inequalities. Suppose some boundary $B$ contains at least two edges from $S_1'$. Then there will be two distinct
edges $e,f\in B\cap S_1'$ such that $e$ and $f$ are oriented in opposite direction along the boundary $B$. 
(Here we see the reason why we partitioned the edge set $S_1$ into the two subsets $S_1'$ and $S_1''$. 
In the plane graph arising from $T$ together with the whole edge set $S_1$ it may happen, that the boundary of a region contains 
at least two edges from $S_1$ and these edges are all oriented in the same direction along the boundary.)
If we remove $e$ and $f$ from
$B$ we get two paths (one of which may be empty) connecting the heads and tails of $e$ and $f$. Now by applying 
the triangle inequality to these two paths and using the 2-optimality condition~(\ref{eq:2-optimality}) for the edges $e$ and $f$ 
we get another inequality for the edges in $S_1'$. We call this inequality the \emph{combined 2-optimality condition} as
it arises by applying the 2-optimality condition in combination with the triangle inequality to edges from both tours $T$ and $S$.
   
In the following we want to apply the combined triangle inequality and the combined 2-optimality condition to neighboring regions 
of the plane
graph formed by the edge set $T\cup S_1'$. This part of the proof is independent of the embedding induced by the 
point coordinates in $\mathbb{R}^2$. In particular this part of the proof can also be applied to point sets in higher 
dimension by choosing an arbitrary planar embedding of the tour $T$. We therefore reduce in the following the problem of bounding 
the total length of all edges in $S_1'$ to a purely combinatorial problem in weighted arborescences. \bigskip

We now give a formal description of the reduction. Consider the plane graph obtained from $T$ together with the edge set $S_1'$. 
Let $H$ be the graph that is obtained from
the geometric dual of this plane graph by removing the vertex corresponding to the outer region. Then each edge in $H$ is a dual 
of an edge in 
$S_1'$. As each edge in $S_1'$ is a chord in the polygon $T$ we know that each edge in $H$ is a cut edge and thus $H$ is a tree. 
See Figure~\ref{fig:DualTree} for an example.

\begin{figure}[t]
\centering
\begin{tikzpicture}[scale=0.25, fill=gray]

\coordinate  (P1) at (32,  5);
\coordinate  (P2) at (18, 18);
\coordinate  (P3) at (27,  1);
\coordinate  (P4) at (15,  1);
\coordinate  (P5) at (23, 13);
\coordinate  (P6) at (40, 13);
\coordinate  (P7) at (34,  1);
\coordinate  (P8) at (15, 11);
\coordinate  (P9) at (11,  3);
\coordinate (P10) at (32, 11);
\coordinate (P11) at (35, 19);
\coordinate (P12) at ( 5, 10);
\coordinate (P13) at (21,  3);
\coordinate (P14) at (29, 19);
\coordinate (P15) at (25,  7);
\coordinate (P16) at (40,  7);
\coordinate (P17) at (36, 15);
\coordinate (P18) at (10, 16);
\coordinate (P19) at ( 1, 20);
\coordinate (P20) at (32, 16);
\coordinate (P21) at ( 9,  7);
\coordinate (P22) at (28, 14);
\coordinate (P23) at (36,  6);
\coordinate (P24) at ( 2,  5);
\coordinate (P25) at ( 7,  1);
\coordinate (P26) at (40,  1);
\coordinate (P27) at ( 1, 14);
\coordinate (P28) at ( 9, 20);
\coordinate (P29) at (20, 10);
\coordinate (P30) at (17,  6);
\coordinate (P31) at ( 6, 15);
\coordinate (P32) at (14, 15);
\coordinate (P33) at (22, 20);
\coordinate (P34) at (29,  8);
\coordinate (P35) at ( 1,  9);
\coordinate (P36) at (40, 20);
\coordinate (P37) at ( 9, 12);
\coordinate (P38) at (14, 20);
\coordinate (P39) at (37, 10);
\coordinate (P40) at ( 1,  1);
\coordinate (P41) at (19, 14);
\coordinate (P42) at ( 5, 19);

\filldraw[fill=black!10!white, draw=red, line width = 1.5] 
(P26) -- (P16) -- (P23) -- (P39) --  (P6) -- (P17) -- (P36) -- (P11) -- (P20) -- (P10) -- 
(P22) -- (P14) -- (P33) --  (P2) -- (P41) --  (P5) -- (P29) --  (P8) -- (P32) -- (P38) -- 
(P28) -- (P42) -- (P19) -- (P27) -- (P31) -- (P18) -- (P37) -- (P21) -- (P12) -- (P35) -- 
(P24) -- (P40) -- (P25) --  (P9) --  (P4) -- (P30) -- (P13) --  (P3) -- (P15) -- (P34) -- 
 (P1) --  (P7) -- cycle;

\foreach \i in {1,...,42}
  \node[circle, minimum size = 2.5mm, inner sep = 0mm] (N\i) at (P\i) {};

\node[circle, minimum size = 2.5mm, inner sep = 0mm, label=right:\hspace*{-1mm}$x_0$] (N26) at (P26) {};
\node[circle, minimum size = 2.5mm, inner sep = 0mm, label=left:$y_0$\hspace*{-1.6mm}]  (N23) at (P23) {};

\tikzstyle{arrow}=[Straight Barb[length=2mm]] 

% solid arrowed line  
\tikzstyle{sal}=[-{Straight Barb[length=1.5mm]}, color=blue!70!white, line width = 1.5] 

\draw[sal] (N26)  to node[below] {$e_0~$} (N23);
\draw[sal]  (N20) -- (N17); 
\draw[sal]  (N20) -- (N17);
\draw[sal]  (N17) -- (N39); 
\draw[sal]  (N34) -- (N22); 
\draw[sal]  (N15) -- (N29); 
\draw[sal]  (N30) --  (N8); 
\draw[sal]  (N41) -- (N33); 
\draw[sal]  (N18) -- (N28); 
\draw[sal]  (N25) -- (N21);

\coordinate  (D1) at ( 38.667,  4.667); % 16 23 26
\coordinate  (D2) at ( 33.600,  8.700); % 1 7 10 17 20 22 23 26 34 39
\coordinate  (D3) at ( 37.667, 12.667); % 6 17 39
\coordinate  (D4) at ( 35.750, 17.500); % 11 17 20 36
\coordinate  (D5) at ( 25.375, 13.125); % 5 14 15 22 29 33 34 41  % x + 1
\coordinate  (D6) at ( 19.667, 17.333); % 2 33 41
\coordinate  (D7) at ( 20.833,  6.333); % 3 8 13 15 29 30
\coordinate  (D8) at ( 12.300, 11.100); % 4 8 9 18 21 28 30 32 37 38 
\coordinate  (D9) at (  5.333, 17.333); % 18 19 27 28 31 42 
\coordinate (D10) at (  4.167,  5.500); % 12 21 24 25 35 40 

\foreach \i in {1,...,10}
  \node[circle, minimum size = 3mm, inner sep = 0mm] (ND\i) at (D\i) {};
  
% solid arrowed dual line  
\tikzstyle{sadl}=[-{Straight Barb[length=1.5mm]}, color=green!80!black, line width = 2] 

\draw[sadl] (ND1) .. controls (35.0, 1.5) .. (ND2);
\draw[sadl] (ND2) .. controls (35.5, 12.5) .. (ND3);
\draw[sadl] (ND2) .. controls (34.0, 15.5) .. (ND4);
\draw[sadl] (ND2) -- (ND5);
\draw[sadl] (ND5) -- (ND6);
\draw[sadl] (ND5) -- (ND7);
\draw[sadl] (ND7) .. controls (16.0,  8.5) .. (ND8);
\draw[sadl] (ND8) .. controls (11.5, 18.0) .. (ND9);
\draw[sadl] (ND8) .. controls (10.0,  4.0) .. (ND10);

\foreach \i in {1,...,42}
  \fill[black] (P\i) circle (3mm);

\foreach \i in {1,...,10}
  \fill[color=black] (D\i) circle (4.5mm);

\end{tikzpicture}
\caption{The arborescence (green edges) in the dual of the plane graph formed by the edges of an optimal tour (red edges)
and the edges in the set $S_1'$ (blue edges) with respect to the edge $e_0=(x_0,y_0)$.} 
\label{fig:DualTree}
\end{figure}
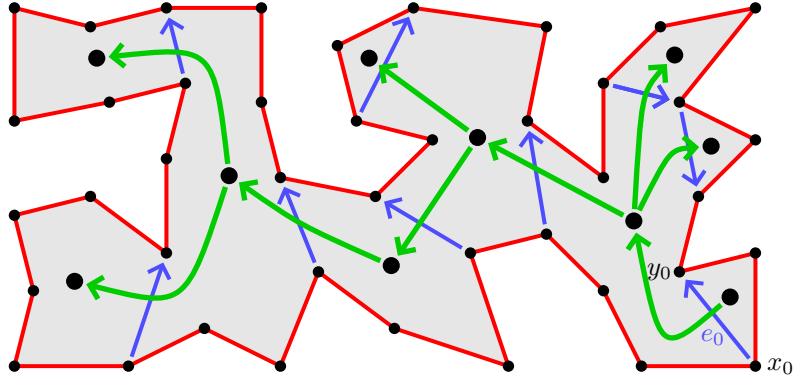

We now want to orient the edges of $H$ to get an arborescence.
An \emph{arborescence} $A=(V,E)$ is a connected directed acyclic graph that satisfies $|\delta^-(x)| \le 1$ for all $x\in V(A)$.
Each arborescence has exactly one \emph{root} $r$ which is the unique vertex $r\in V(A)$ with $\delta^-(r) = \emptyset$. 
For an arborescence $A$ with root $r$ we say that $A$ \emph{is rooted at} $r$.

The set $S_1'$ has been defined with respect to some edge $e_0=(x_0,y_0)$. 
The tree $H$ contains a vertex that corresponds to the region of the plane graph that is bounded by the edge $e_0$ and 
the edges in $E(T) \setminus  T_{[x_0,y_0]}$. By choosing this vertex as the root and orienting all edges in $H$ from the root to the leaves,
we get an arborescence $A$ from the tree $H$ (see Figure~\ref{fig:DualTree}).

We want to define two weight functions on the edge set $E(A)$ of the arborescence $A$ to capture the weights of the edges in $T$ and the
edges in $S_1'$. 
First we define the function $c:E(A) \to \mathbb{R}_{>0}$ to be the weight of the corresponding dual edge in $S_1'$.
Secondly, we define  a weight function $w:E(A) \to \mathbb{R}_{>0}$ as follows. Let $e=(x,y)$ be a directed edge in $E(A)$.
Let $Y$ be the region corresponding to the vertex $y$ in the plane graph formed by $E(T) \cup S'_1$. 
Then we define $w(e)$ to be the weight of all the edges in $E(T)$ that belong to the boundary of $Y$.

For the arborescence $A$ and the two weight functions $c$ and $w$ we can now state the above mentioned combined triangle inequality 
and the combined 2-optimality condition as follows:

\begin{lemma}\label{lemma:main-property-of-S1'}
Let $V\subseteq \mathbb{R}^2$ be a Euclidean TSP instance with distance function $\overline c: V\times V \to \mathbb{R}$ and
let $T$ be an optimal tour. Let $S$ be a 2-optimal tour such that $S$ and $T$ are crossing-free. 
Let $S_1'$ be defined as in Section~\ref{sec:edge-partition} with respect to some edge $e_0=(x_0,y_0)$. 
Let $A$ be the arborescence derived from the geometric dual of the plane graph $T\cup S_1'$ with 
weight functions $w:E(A) \to \mathbb{R}_{>0}$
and $c:E(A) \to \mathbb{R}_{>0}$ as defined above. Then 
we have
\begin{equation}
c(e) ~\le~ w(e) + \sum_{f\in \delta^+(y)} c(f) ~~~\mbox{for all $e= (x,y)\in E(A)$} .
\label{eq:combined-triangle-inequality}
\end{equation}
and 
\begin{equation}
c(x,y) + c(y,z) ~\le~ w(x,y) + \sum_{f\in \delta^+(y)\setminus\{(y,z)\}} c(f) \mbox{~~~~~ for all } (x,y), (y,z) \in E(A) 
\label{eq:combined-2-optimality}
\end{equation} 
\end{lemma}

\begin{proof}
Let $f=(a,b)$ be an edge in $S_1'$ and $f'=(a',b')$ its corresponding dual edge in $A$. By definition we have $c(f') = \overline c(f)$. 
The vertex $b'$ corresponds to a region $R$ in the plane graph on $V$ with edges $E(T) \cup S'_1$. 
By the triangle inequality the length $\overline c(f)$ of the edge $f$ is bounded by the length of all other edges in the boundary of the region $R$. 
Using the definitions of the functions $c$ and $w$ we therefore get: 
\[ c(f') ~=~ \overline c(f)  ~\le~ \sum_{g\in R\cap E(T)} \overline c(g) +  \sum_{g\in (R\cap S'_1) \setminus\{f\}} \overline c(g) 
         ~=~ w(f') + \sum_{g\in \delta^+(b')} c(g).\]
This proves condition~(\ref{eq:combined-triangle-inequality}).

We now prove property (\ref{eq:combined-2-optimality}). Let $f\in S'_1$. By definition of the set $S'_1$ we know that $S'_1$ is defined with 
respect to an edge $e_0=(x_0,y_0) \in S_1$ and the $x_0$-$y_0$-path $T_{[x_0,y_0]}$ contains the endpoints of all other edges in $S_1$. 
Let $\phi : V(T_{[x_0,y_0]}) \to\mathbb{N}$ such that $\phi(z)$ for $z \in  V(T_{[x_0,y_0]})$ denotes the distance 
(in terms of the number of edges) between $x_0$ and $z$ in $T_{[x_0,y_0]}$. The definition of the set $S'_1$ implies that $\phi(a) < \phi(b)$
for each edge $(a,b) \in S'_1$. Each edge in $S'_1$ can be seen as a shortcut for the path $T_{[x_0,y_0]}$. 
For an edge $f=(a,b)\in S'_1$ with dual edge $f'=(a',b') \in E(A)$ we denote by $\left(\delta^+(b')\right)'$
all edges dual to the edges in $\delta^+(b')$. The edge $f=(a,b)$ and the edges in $\left(\delta^+(b')\right)'$
belong to the border
of a region of the graph on $V$ with edge set $E(T) \cup S'_1$. Along this border the edge $f=(a,b)$ is directed opposite to all edges 
in $\left(\delta^+(b')\right)'$. Therefore, the triangle inequality together with the 2-optimality condition~(\ref{eq:2-optimality})
for the set $S'_1$ imply for each edge $(u,v)\in \left(\delta^+(b')\right)'$:
\[ \overline c(a,b) + \overline c (u,v) ~\le~ w(a',b') + \sum_{g\in \left(\delta^+(b')\right)' \setminus \{(u,v)\}} \overline c(g) \]
We have $\overline c(a,b) = c(a',b')$ and $\overline c(u,v) = c(b', x')$ for the vertex $x'\in V(A)$ such that $(u,v)\in S_1$ 
is the dual edge to $(b', x')\in E(A)$. 
Therefore we get:
\[ c(a',b') + c (b',x') ~\le~ w(a',b') + \sum_{g\in \delta^+(b') \setminus \{(b',x')\}} c(g) \]
\end{proof}

We call condition~(\ref{eq:combined-triangle-inequality}) the \emph{combined triangle inequality} and 
condition~(\ref{eq:combined-2-optimality}) the \emph{combined 2-optimality condition}.
Note that these two conditions can be formulated for any arborescence $A$ with weight functions $c$ and $w$. 
In the next section we will show that if these two conditions are satisfied for an arborescence $A$ then we can bound  $c(A)/w(A)$
by $O(\log(|E(A)|) / \log \log(|E(A)|))$.

\section{The Arborescence Lemmas}
\label{sec:arborescence-lemmas}

Let $A$ be an arborescence with weight functions $w:E(A) \to \mathbb{R}_{>0}$
and $c:E(A) \to \mathbb{R}_{>0}$. We will first show that 
the combined triangle inequality~(\ref{eq:combined-triangle-inequality}) and the 
combined 2-optimality condition~(\ref{eq:combined-2-optimality}) imply two additional properties.
For this we need the following definition. For an arborescence $A=(V,E)$ and an edge $e = (x,y)\in E(A)$ we denote by
$A_e$ the sub-arborescence rooted at $x$ that contains the edge $e$ and all descendants of $y$, see Figure~\ref{fig:sub-arborescence} for an example.

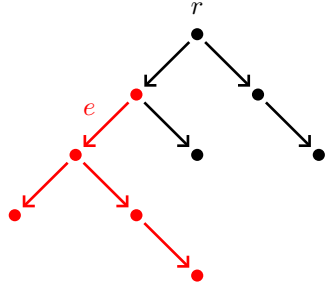
\begin{figure}[t]
\centering
\begin{tikzpicture}[scale=0.8]
\tikzstyle{vertex}=[black,circle,minimum size=8,inner sep=0]
\tikzstyle{arrow}=[-{Straight Barb[length=1.0mm]}]

\node[vertex, label = above:$r$]      (A)  at ( 3  , 4) {};
\node[vertex, red] (B)  at ( 2  , 3) {};
\node[vertex]      (C)  at ( 4  , 3) {};
\node[vertex, red] (D)  at ( 1  , 2) {};
\node[vertex]      (E)  at ( 3  , 2) {};
\node[vertex]      (F)  at ( 5  , 2) {};
\node[vertex, red] (G)  at ( 0  , 1) {};
\node[vertex, red] (H)  at ( 2  , 1) {};
\node[vertex, red] (I)  at ( 3  , 0) {};

\fill (A) circle (1mm);
\fill[red] (B) circle (1mm);
\fill (C) circle (1mm);
\fill[red] (D) circle (1mm);
\fill (E) circle (1mm);
\fill (F) circle (1mm);
\fill[red] (G) circle (1mm);
\fill[red] (H) circle (1mm);
\fill[red] (I) circle (1mm);

 \draw[arrow,  line width=1] (A) to (B);
 \draw[arrow,  line width=1] (A) to (C);
 \draw[arrow,  line width=1] (B) to (E);
 \draw[arrow,  line width=1] (C) to (F);
 
 \draw[arrow, red, line width=1] (D) to (G);
 \draw[arrow, red, line width=1] (D) to (H);
 \draw[arrow, red, line width=1] (H) to (I);
 \draw[arrow, red, line width=1] (B) to node[above left] {$e$} (D);
\end{tikzpicture}
\caption{An arborescence $A$ with root $r$. Shown in red is the sub-arborescence $A_e$ defined by the edge $e$.}
\label{fig:sub-arborescence}
\end{figure}

\begin{lemma}\label{lemma:weight-bound}
Let $A$ be an arborescence with weight functions $w:E(A) \to \mathbb{R}_{>0}$
and $c:E(A) \to \mathbb{R}_{>0}$ that satisfies  
the combined triangle inequality~(\ref{eq:combined-triangle-inequality}). 
Then we have 
\begin{equation}
c(e) ~\le~ w(A_e) ~~~\mbox{for all $e\in E(A)$}
\label{eq:weight-bound}
\end{equation}
\end{lemma}

\begin{proof}
This follows by induction on the height of the sub-arborescence $A_e$.
If $e=(x,y)$ is an edge in $E(A)$ such that $y$ is a leaf then the combined triangle inequality~(\ref{eq:combined-triangle-inequality})
implies $c(e) \le w(e) = w(A_e)$. 
For an arbitrary edge $e=(x,y)\in E(A)$ we get by induction: 
$$ c(e) ~\le~ w(e) + \sum_{f\in \delta^+(y)} c(f) ~\le~ w(e) + \sum_{f\in \delta^+(y)} w(A_f) ~=~ w(A_e) .$$
\end{proof}

In the following lemma we have a statement about the maximum weight of a possibly empty edge set. As usual we assume 
$\max \emptyset = -\infty$.

\begin{lemma}\label{lemma:weight-bound-2}
Let $A$ be an arborescence with weight functions $w:E(A) \to \mathbb{R}_{>0}$
and $c:E(A) \to \mathbb{R}_{>0}$ that satisfies  
the combined 2-optimality condition~(\ref{eq:combined-2-optimality}). Then we have:
\begin{equation}
2 \cdot \max_{f\in \delta^+(y)} c(f) ~\le~ w(x,y) -c(x,y) + \sum_{g\in \delta^+(y)} c(g) ~~~\mbox{for each edge } (x,y) \in E(A)
\label{eq:weight-bound-2}
\end{equation} 
\end{lemma}

\begin{proof} 
From the combined 2-optimality condition~(\ref{eq:combined-2-optimality}) we get
$$2 \cdot c(y,z) ~\le~ w(x,y) - c(x,y) + \sum_{f\in \delta^+(y)} c(f) \mbox{~~~~~ for all } (x,y), (y,z) \in E(A) . $$
As the right hand side of this inequality is independent of the edge $(y,z)$ we can therefore replace $c(y,z)$ by the term
$\max_{f\in \delta^+(y)} c(f)$ on the left hand side.
\end{proof}

Our next goal is to bound $c(A)$ in terms of $w(A)$ for arborescences $A$ satisfying~(\ref{eq:combined-triangle-inequality}) 
and~(\ref{eq:combined-2-optimality}). The following two lemmas prove such a statement for certain subsets of $E(A)$.

\begin{lemma}
Let $A=(V,E)$ be an arborescence with weight functions $w:E(A)\to \mathbb{R}_{> 0}$ and $c:E(A)\to \mathbb{R}_{> 0}$
that satisfies the combined 2-optimality condition~(\ref{eq:combined-2-optimality}).
 For some fixed number $k \in \mathbb{R}_{>0}$ define 
$E' := \{ (x,y) \in E(A): \max_{f\in \delta^+(y)} c(f) > \frac{1}{k} \cdot c(x,y) \}$.
Then we have: \[ c(E') ~\le~ \frac{k}{2} \cdot w(A) .\]
\label{lemma:E'}
\end{lemma}
\begin{proof}
By definition of $E'$ and using Lemma~\ref{lemma:weight-bound-2} we have for each edge $(x,y)\in E'$:
\[ \frac{2}{k}\cdot c(x,y) ~<~ 2\cdot \max_{f\in \delta^+(y)}c(f) \le w(x,y) - c(x,y) + \sum_{g\in \delta^+(y)} c(g) .\]
Adding this inequality for all $(x,y)\in E'$ and using that the left hand side and therefore also the right hand side of 
inequality~(\ref{eq:weight-bound-2}) is non-negative
we get:
\begin{eqnarray*}
\frac{2}{k}\cdot \sum_{(x,y)\in E'} c(x,y) & < &  \sum_{(x,y)\in E'}\left( w(x,y) - c(x,y) + \sum_{g\in \delta^+(y)} c(g)\right) \\
                                           &\le&  \sum_{(x,y)\in E(A)}\left( w(x,y) - c(x,y) + \sum_{g\in \delta^+(y)} c(g)\right) \\
                                           & \le & w(A)
\end{eqnarray*} 
\end{proof}

For a fixed number $k\in\mathbb{R}_{>0}$ and a number $r\in \mathbb{R}_{\ge 0}$ we define the edge set $E_r \subseteq E(A)$ as follows:
\begin{equation}
E_r ~:=~ \left\{ e=(x,y)\in E(A): r < c(e) \le \frac{k}{4} \cdot r \mbox{ and } c(f) \le \frac{1}{k}\cdot c(e) ~\forall f\in \delta^+(y) \right\}
\label{def:Er}
\end{equation}

\begin{lemma}\label{lemma:arborescence}
Let $A=(V,E)$ be an arborescence with weight functions $w:E(A)\to \mathbb{R}_{> 0}$ and $c:E(A)\to \mathbb{R}_{> 0}$
that satisfies the combined triangle inequality~(\ref{eq:combined-triangle-inequality})
and the combined 2-optimality condition~(\ref{eq:combined-2-optimality}).
Let $E_r$ be defined as in~(\ref{def:Er}).
Then we have: \[c(E_r) ~\le~ 2\cdot w(A) .\]  
\end{lemma}

\begin{proof}
Let $e = (x,y) \in E_r$. We first prove by induction on the cardinality of $E(A_e) \cap E_r$: 
\begin{equation}
w(A_e) ~\ge ~ c(e) + \sum_{f\in \left(E(A_e) \cap E_r\right) \setminus \{e\}} \left ( c(f) - \frac{r}{4}\right)
\label{eq:induction}
\end{equation}
 
If $|E(A_e) \cap E_r| = 1$ then $E(A_e) \cap E_r = \{e\}$ and therefore $(E(A_e) \cap E_r) \setminus \{e\} = \emptyset$.
Inequality~(\ref{eq:induction}) then states $w(A_e) \ge  c(e)$ which holds because of Lemma~\ref{lemma:weight-bound}.

Now assume that $|E(A_e) \cap E_r| > 1$ and that inequality~(\ref{eq:induction}) holds for all edges $f \in E_r$ with $|E(A_f) \cap E_r| < |E(A_e) \cap E_r|$. 
From the definition of the set $E_r$ we get for each edge $f\in \delta^+(y)$: 
\begin{equation}
c(f) ~\le ~ \frac{1}{k}\cdot c(e)  ~\le~ \frac{1}{k}\cdot \frac{k}{4} \cdot r ~=~  \frac{r}{4}
\label{eq:r/4-bound}
\end{equation}
We define the following two sets of edges: 
\begin{equation*}
X ~ := ~ \{f\in \delta^+(y): E(A_f) \cap E_r = \emptyset\}
\end{equation*}
and
\begin{equation*}
F ~ := ~ \{f\in (E_r\cap E(A_e)) \setminus \{e\}: \mbox{ no edge } h\in E_r \mbox{ lies on a path from $f$ to $e$ in $A$} \}
\end{equation*}
For each edge $f\in \delta^+(y)$ we either have $E(A_f) \cap E_r = \emptyset$ or $E(A_f) \cap E_r \not= \emptyset$. In the first case the edge $f$ belongs to 
the set $X$.
In the second case at least one edge from $A_f$ belongs to $F$. Thus we have
\begin{equation}
|F| + |X| ~\ge~ |\delta^+(y)| ~~~\Rightarrow ~~~ |\delta^+(y) \setminus X| ~\le~ |F|. 
\label{eq:setsizes}
\end{equation}
 
By the induction hypothesis, inequality~(\ref{eq:induction}) holds for each edge $f\in F$ and we can now prove inequality~(\ref{eq:induction}) for the edge~$e$:
\begin{eqnarray*}
w(A_e) & \ge                                            & \sum_{f\in F} w(A_f) + \sum_{f\in X} w(A_f) + w(e) \\
       & \stackrel{(\text{\ref{eq:weight-bound}})}{\ge}  & \sum_{f\in F} w(A_f) + \sum_{f\in X} c(f) + w(e) \\
       & \stackrel{(\text{\ref{eq:combined-triangle-inequality}})}{\ge} & \sum_{f\in F} w(A_f) + c(e) - \sum_{f\in \delta^+(y)\setminus X} c(f) \\
       & \stackrel{(\text{\ref{eq:r/4-bound}})}{\ge}    & \sum_{f\in F} w(A_f) + c(e) - \sum_{f\in \delta^+(y)\setminus X} \frac{r}{4} \\
       & \stackrel{(\text{\ref{eq:setsizes}})}{\ge}     & \sum_{f\in F} \left( w(A_f) - \frac{r}{4} \right)  + c(e)  \\
       & \stackrel{(\text{\ref{eq:induction}})}{\ge}    & \sum_{f\in F} \left( c(f) + \sum_{g\in (E(A_f)\cap E_r)\setminus\{f\}} \left( c(g) -\frac{r}{4} \right) - \frac{r}{4} \right)  + c(e)  \\
       &                           =                    & c(e) + \sum_{f\in (E(A_e) \cap E_r) \setminus \{e\}} \left ( c(f) - \frac{r}{4}\right)  
\end{eqnarray*} 
The last equality holds because each edge in  $(E(A_e) \cap E_r) \setminus \{e\}$ appears exactly once in the sets 
$(E(A_f)\cap E_r)$ for $f\in F$.

By definition of the set $E_r$ we have for each edge $f\in E_r$: 
\begin{equation}
c(f) ~\ge~ r ~~~\Rightarrow~~~ \frac{1}{2} \cdot c(f) ~\ge~ \frac{r}{4} ~~~\Rightarrow~~~ c(f) - \frac{r}{4} ~\ge~ \frac{1}{2} \cdot c(f)
\label{eq:Er-edge-cost-lower-bound}
\end{equation} 
 
Inequality~(\ref{eq:induction}) therefore implies 
\begin{eqnarray*}
w(A_e) & \ge   & c(e) + \sum_{f\in (E(A_e) \cap E_r) \setminus \{e\}} \left ( c(f) - \frac{r}{4}\right)\\
       & \ge   & \sum_{f\in E(A_e) \cap E_r} \left ( c(f) - \frac{r}{4}\right) \\
       & \stackrel{(\text{\ref{eq:Er-edge-cost-lower-bound}})}{\ge}   & \sum_{f\in E(A_e) \cap  E_r} \left ( \frac{1}{2}  \cdot c(f)\right)\\
       & =     & \frac{1}{2} \cdot c(E(A_e) \cap E_r)  
\end{eqnarray*} 

Now choose a minimal set of edges $e_1, e_2, \ldots \in E_r$ such that $E_r \subseteq \bigcup_{i} E(A_{e_i})$.
Then 
\[ w(A) ~\ge~ w(\bigcup_{i} E(A_{e_i})) ~=~ \sum_{i} w(E(A_{e_i})) ~\ge~ \frac{1}{2} \cdot \sum_{i} c(E(A_{e_i}) \cap E_r)  ~=~ \frac{1}{2} \cdot c(E_r)\]
\end{proof}

\begin{lemma}
Let $A=(V,E)$ be an arborescence with weight functions $w:E(A)\to \mathbb{R}_{> 0}$ and $c:E(A)\to \mathbb{R}_{> 0}$
that satisfies the combined triangle inequality~(\ref{eq:combined-triangle-inequality})
and the combined 2-optimality condition~(\ref{eq:combined-2-optimality}).
Moreover we assume that $c(A) \ge 18 \cdot w(A)$.
Then we have: \[c(A) ~\le~ 12\cdot \frac{\log(|E(A)|)}{\log \log (|E(A)|)}\cdot w(A) .\]
\label{lemma:mainlemma}
\end{lemma}

\begin{proof}
We define $k := c(A) / w(A)$. By assumption we have $k \ge 18$.
For $i= 1, 2, \ldots, \lfloor k/6 \rfloor$ we define $r_i := \left( \frac{4}{k}\right)^i\cdot w(A)$ and for these numbers we 
define sets $E_{r_i}$ as in~(\ref{def:Er}). By Lemma~\ref{lemma:weight-bound} we have 
$c(e) \le w(A) = \frac{k}{4} \cdot \left(\frac{4}{k}\right)^1\cdot w(A) = \frac{k}{4}\cdot r_1$ and therefore we have:
\[\bigcup_{i=1}^{\lfloor k/6 \rfloor} E_{r_i} = \left\{ e=(x,y)\in E(A): \left(\frac{4}{k}\right)^{\lfloor k/6 \rfloor} \cdot w(A) < c(e) \mbox{ and } c(f) \le \frac{1}{k}\cdot c(e) ~\forall f\in \delta^+(y) \right\} .\]
Define \[E' := \{ (x,y) \in E(A): \max_{f\in \delta^+(y)} c(f) > \frac{1}{k} \cdot c(x,y) \}\] and 
\[E^* := \{e\in E(A): c(e) \le \left(\frac{4}{k}\right)^{\lfloor k/6 \rfloor} \cdot w(A)\} .\]
Then we have \[E(A) = E' \cup E^* \cup \bigcup_{i=1}^{\lfloor k/6 \rfloor} E_{r_i}.\] 

Using Lemma~\ref{lemma:E'} and Lemma~\ref{lemma:arborescence} we get:
\begin{eqnarray*}
k\cdot w(A) ~ = ~ c(A) & \le & \sum_{i=1}^{\lfloor k/6 \rfloor} c(E_{r_i}) + c(E') + c(E^*) \\
                       & \le & {\lfloor k/6 \rfloor} \cdot 2\cdot w(A) + \frac{k}{2} \cdot w(A) + \left(\frac{4}{k}\right)^{\lfloor k/6 \rfloor} \cdot w(A) \cdot|E^*| \\
                       & \le &  \frac{5}{6} \cdot k \cdot w(A) + \left(\frac{4}{k}\right)^{\lfloor k/6 \rfloor} \cdot w(A) \cdot|E^*|
\end{eqnarray*}

This implies 
\begin{equation}
|E(A)| ~\ge~ |E^*| ~\ge~ \frac{k}{6} \cdot \left(\frac{k}{4}\right)^{\lfloor k/6 \rfloor} ~\ge~ \left(\frac{k}{6}\right)^{k/6} .
\label{eq:sizeE}
\end{equation}

The function $\frac{\log x}{\log \log x}$ is monotone increasing for $x > 18$. Therefore we get from inequality~(\ref{eq:sizeE}): 

\begin{eqnarray*}
2\cdot \frac{\log(|E(A)|)}{\log\log(|E(A)|)}\cdot w(A) 
& \ge & 2\cdot \frac{\log\left(\left(\frac{k}{6}\right)^{k/6}\right)}{\log\log\left(\left(\frac{k}{6}\right)^{k/6}\right)}\cdot w(A)\\
& = & 2\cdot \frac{\frac{k}{6} \cdot \log\left(\frac{k}{6}\right)}{\log\left(\frac{k}{6}\right) + \log\log\left(\frac{k}{6}\right)}\cdot w(A)\\ 
& \ge & \frac{k}{6} \cdot w(A)\\
& = & \frac{1}{6} \cdot c(A)
\end{eqnarray*}

\end{proof}

\section{Proof of Theorem~\ref{thm:main-crossingfree}}
\label{sec:proof}

Lemma~\ref{lemma:main-property-of-S1'} in combination with Lemma~\ref{lemma:mainlemma} shows that we can bound the length of
all edges in $S_1'$ by $O(\log n/\log \log n)$ times the length of an optimal tour $T$.
The statement of Lemma~\ref{lemma:main-property-of-S1'} also holds for the set $S''_1$: We can define an arborescence almost the same way as we did 
for the set $S'_1$ by taking the dual of the graph on $V$ formed by the edges of $T$ and $S''_1$ without the vertex for the outer region. The
only minor difference is the choice of the root vertex. For $S'_1$ we have chosen as root the vertex that corresponds to the region $R$
bounded by the edge $e_0=(x_0,y_0)$ and the edges in $E(T)\setminus T_{[x_0,y_0]}$. For the arborescence for $S''_1$ we choose as a root the region 
that contains $R$.  The proof of Lemma~\ref{lemma:main-property-of-S1'} then without any changes shows that the statement of
Lemma~\ref{lemma:main-property-of-S1'} also holds for the set $S''_1$. Similarly, by exchanging the role of the outer and the inner region of $T$,
Lemma~\ref{lemma:main-property-of-S1'} also holds for the sets $S'_2$ and $S''_2$. We are now able to prove our main result:\medskip

\noindent
\textit{Proof of Theorem~\ref{thm:main-crossingfree}:~}
Let $V\subseteq \mathbb{R}^2$ with $|V| = n$ be a non-degenerate Euclidean TSP instance, 
$T$ an optimal tour for $V$ and $S$ a 2-optimal tour for $V$ such that $T$ and $S$ are crossing-free.
We partition the tour $S$ into the five (possibly empty) sets $S'_1$, $S''_1$, $S'_2$, $S''_2$, and $S_3$ as defined in Section~\ref{sec:edge-partition}.
Then $c(S_3) \le c(T)$. We claim that $c(S'_1) = O(\log n / \log \log n) \cdot c(T)$.
If $c(S'_1) < 18 \cdot c(T)$ this is certainly the case. Otherwise by Lemma~\ref{lemma:main-property-of-S1'} and 
Lemma~\ref{lemma:mainlemma} we get 
\[c(S'_1) ~\le~ 12\cdot \frac{\log(n)}{\log \log (n)}\cdot c(T) \]
which again proves the claim. As observed above, Lemma~\ref{lemma:main-property-of-S1'} also holds for the sets $S''_1$, $S'_2$, and $S''_2$. 
Therefore, we can apply the same argument to the sets $S''_1$, $S'_2$, and $S''_2$ and get
\[c(S) ~=~ c(S'_1) + c(S''_1) + c(S'_2) + c(S''_2) + c(S_3) ~=~  O(\log n / \log \log n) \cdot c(T).\]
\qed

\bibliography{Euclidean2OPT} 

\end{document}